\documentclass[preprint,12pt]{elsarticle}
\usepackage{times,amsmath,epsfig}
\usepackage{xspace,latexsym,syntonly}
\usepackage{amssymb}
\usepackage{amsfonts}
\usepackage{textcomp}
\usepackage{graphicx}
\usepackage{dsfont}
\usepackage{enumerate}

\usepackage{lineno}

\biboptions{sort&compress}

\journal{Computer Communications}

\newcommand{\SNR}{\ensuremath{\hbox{SNR}}}

\newcounter{theorem}
\newtheorem{theorem}{Theorem}
\newtheorem{lemma}[theorem]{Lemma}
\newtheorem{claim}[theorem]{Claim}
\newcounter{definition}
\newtheorem{definition}{Definition}

\newcommand{\beq}{\begin{equation}}
\newcommand{\eeq}{\end{equation}}
\newcommand{\bea}{\begin{array}}
\newcommand{\ena}{\end{array}}
\newcommand{\bds}{\begin {itemize}}
\newcommand{\eds}{\end {itemize}}
\newcommand{\bdf}{\begin{definition}}
\newcommand{\blm}{\begin{lemma}}
\newcommand{\edf}{\end{definition}}
\newcommand{\elm}{\end{lemma}}
\newcommand{\bthm}{\begin{theorem}}
\newcommand{\ethm}{\end{theorem}}
\newcommand{\bprp}{\begin{prop}}
\newcommand{\eprp}{\end{prop}}
\newcommand{\bcl}{\begin{claim}}
\newcommand{\ecl}{\end{claim}}
\newcommand{\bcr}{\begin{coro}}
\newcommand{\ecr}{\end{coro}}
\newcommand{\bquest}{\begin{question}}
\newcommand{\equest}{\end{question}}

\newcommand{\larrow}{{\larrow}}

\newenvironment{proof}{\textit{Proof:}}{\hfill$\square$}
\begin{document}

\begin{frontmatter}

\title{Topology management and outage optimization for multicasting over slowly fading multiple access networks\tnoteref{Ack_and_ICCCN}}
\tnotetext[Ack_and_ICCCN]{This research was supported by the Israel Ministry of Labor, Trade and Commerce, as part of the RESCUE consortium. Part of this research was presented in ICCCN 2013 \cite{Zanko_Network_2013Submeeted}}

\author{Avi Zanko\corref{cor1}\fnref{Univaddress}}
\ead{avizanz@gmail.com}
\fntext[Universityaddress]{School of Engineering, Bar-Ilan University, Ramat-Gan, 52900,Israel}
\cortext[cor1]{Corresponding author}

\author{Amir Leshem\fnref{Universityaddress}}
\ead{leshem.amir2@gmail.com}

\author{Ephraim Zehavi\fnref{Universityaddress}}
\ead{ephiz@yahoo.com}

\begin{abstract}
This paper examines the problem of rate allocation for multicasting over slow Rayleigh fading channels using network coding. In the proposed model, the network is treated as a collection of Rayleigh fading multiple access channels. In this model, rate allocation scheme that is based solely on the statistics of the channels is presented. The rate allocation scheme is aimed at minimizing the outage probability. An upper bound is presented for the probability of outage in the fading multiple access channel. A suboptimal solution based on this bound is given. A distributed primal-dual gradient algorithm is derived to solve the rate allocation problem.
\end{abstract}


\end{frontmatter}


\section{Introduction}
\label{sec:introduction}
Network coding extends the functionality of intermediate nodes from storing/forwarding packets to performing algebraic operations on received data. If network coding is permitted, the multicast capacity of a network with a single source has been shown to be equal to the minimal min-cut between the source and each of its destinations \cite{Ahlswede_Network_2000}. In the past decade, the concept of combining data by network coding has been extensively extended by e.g. \cite{Li_Linear_2003,Jaggi_Low_2003,Barbero_Heuristic_2006}
and it is well known that in order to achieve the multicast rate, a linear combination over a finite field suffices if the field size is larger than the number of destinations. Moreover, centralized linear network coding can be designed in polynomial time \cite{Jaggi_Polynomial_2005}. Decentralized linear network coding can be implemented using a random code approach \cite{Ho_A_random_2006}. A comprehensive survey of network coding can be found in e.g., \cite{Fragouli_Network_2007,Ho_Network_2008}.

Many network resource allocation problems can be formulated as a constrained maximization of a certain utility function. The problem of network utility maximization  has been explored extensively in the past few decades \cite{Palomar_A_Tutorial_2006,Chiang_Layering_2007}. We briefly introduce related work on topology management and rate allocation for network coding in multicast over wireless networks. The problem of finding a minimum-cost scheme (while maintaining a certain multicast rate) in coded networks was studied by Lun et al. \cite{Lun_Network_2004,Lun_Minimum_2006}. They showed that there is no loss of optimality when the problem is decoupled into: finding the optimal coding rate allocation vector (also known as subgraph selection) and designing the code that is applied over the optimal subgraph. Moreover, in many cases, optimal subgraphs can be found in polynomial time. If in addition the cost function is also convex and separable, the solution can be found in a decentralized manner, where message passing is required solely between directly connected nodes. This decentralized solution, if coupled with random network coding (e.g. \cite{Lun_On_2008,Chou_Practical_2003}) provides a fully distributed scheme for multicast in coded wireline networks. This has prompted many researchers to develop different algorithms that find minimum-cost rate allocation solutions distributively; e.g. \cite{Cui_Optimal_2004,Bhadra_Min_2006,Wu_Distributed_2006, Xi_Distributed_2010}.

When addressing the problem of rate allocation for multicast with network coding in wireless networks, Lun et al., \cite{Lun_Minimum_2006,Lun_Achieving_2005} tackled the problem through the so-called \textit{wireless multicast advantage} phenomenon. This phenomenon simply comes down to the fact that when interference is avoided in the network (e.g., by avoiding simultaneous transmissions), communication between any two nodes is overheard by their nearby nodes due to the broadcast nature of the wireless medium. In \cite{Lun_Achieving_2005}, the wireless multicast advantage was used to reduce the transmission energy of the multicast scheme (since when two nodes communicate, some of their nearby nodes get the packet for "free"). Therefore, their wireline minimum-cost optimization problem was updated accordingly \citep[see][eq.(1) and (40)]{Lun_Achieving_2005}. In \cite{Xi_Distributed_2010} interference is allowed but is assumed to be limited. Joint optimal power control, network coding and congestion control is presented for the case of very high SINR (signal to noise plus interference ratio). This interference assumption implies that there are some limitations on simultaneous transmissions and this is taken into account in the optimization problem. In \cite{Yuan_A_Cross_2006} the problem of joint power control, network coding and rate allocation was studied. They showed that the throughput maximization problem can be decomposed into two parts: subgraph selection at the network layer and power control at the physical layer. A primal dual algorithm was given that converges to the optimal solution provided that the capacity region is convex with respect to the power control variables (i.e., when interference are ignored). On the other hand, to take interference into account a game theoretic method was derived to approximately characterize the capacity region.

In wireless networks, it is reasonable to assume that there is no simultaneous packet transmission or reception by any transceiver. These properties of the wireless medium introduced a new cross-layer interactions that may not exist in the wired network. Sagduyu et al. \cite{Sagduyu_On_Joint_2007} analyzed and designed wireless network codes in conjunction with conflict-free transmission schedules in wireless ad hoc networks. They studied the cross-layer design possibilities of joint medium access control and network coding. It was shown that when certain objectives such as throughput or delay efficiency are considered, then network codes must be jointly designed with medium access control. The joint design of medium access control and network coding \cite{Sagduyu_On_Joint_2007} was formulated as a nonlinear optimization problem. In \cite{Niati_Throughput_2012} the work reported in \cite{Sagduyu_On_Joint_2007} was extended and a linear formulation was derived.

However, there are certain other considerations that must be taken into account in the search for a rate allocation vector in wireless networks. The wireless medium varies over time and suffers from fading channels due to multipath or shadowing, for example. In \cite{Ozarow_Information_1994} the block fading model was introduced. In this model the channel gain is assumed to be constant over each coherence time interval. Typically, fading models are classified as fast fading or slow fading. In fast fading, the coherence time of the channel is small relative to a code block length and as a consequence the channel is ergodic with a well-defined Shannon capacity (also known as the ergodic capacity \cite{Goldsmith_Capacity_1997}). In slow fading, the code block length and the coherence time of the channel are of the same order. Hence, the channel is not ergodic and the Shannon capacity is not usually a good measure of performance. The notion of outage capacity was introduced in \cite{Ozarow_Information_1994} for transmitting over fading channels when the channel gain is available only at the receiver. In this approach, transmission takes place at a certain rate and tolerates some information loss when an outage event occurs. An outage event occurs whenever the transmitted rate is not supported by the instantaneous channel gain; i.e., when the channel gain is too low for successful decoding of the transmitted message. It is assumed that outage events occur with low probability that reliable communication is available most of the time. A different strategy to deal with slow fading is the broadcast channel approach \cite{Shamai_A_broadcast_1997}. In this approach different states of the channel are treated as channels toward different receivers (a receiver for each state). Hence, the same strategy as used for sending common and private messages to different users on the Gaussian broadcast channel can be applied here. When the channel gain is also available at the encoder, the encoder can adapt the power and the transmission rate as a function of the instantaneous state of the channel and thus can achieve a higher rate on average. Moreover, as regards the outage capacity, the transmitter can use power control to conserve power by not transmitting at all during designated outage periods.

When dealing with outage capacity for fading MAC, the common outage has a similar definition to the outage event in the point to point case. A common outage event is declared whenever we transmit with rates that are not supported by the instantaneous channel gains. If the channel gains are available at both the decoder and the encoders, additional notions of capacities for the fading MAC need to be taken into account. The throughput capacity region for the Gaussian fading MAC was introduced in \cite{Tse_Multiaccess_1998}. In a nutshell, this is the Shannon capacity region where the codewords can be chosen as a function of the realization of the fading with arbitrarily long coding delays. However, as for the point to point case, this approach is not realistic in slow fading cases since it requires a very long delay to average out the fading effect. \cite{Hanly_Multiaccess_1998} derived the delay limited capacity for the Gaussian fading MAC (also known as the zero outage capacity). In the delay limited capacity, unlike the throughput capacity, the chosen coding delay has to work uniformly for all fading processes with a given stationary distribution. However, the delay limited capacity is somewhat pessimistic due to the demand to maintain a constant rate under any fading condition. The outage capacity region and the optimal power allocation for a fading MAC were described in \cite{Li_Outage_2005}. As was pointed out in \cite{Li_Outage_2005}, in a slow fading environment, the decoding delay depends solely on the code-length employed and not on the time variation of the channel.

The demand for interference free channels at all nodes means that some level of orthogonality is required between different transmissions in the network. Avoiding interference between all nodes comes at the cost of loss of expensive
bandwidth, or alternatively leads to rate degradation in band
limited systems. The same argument can be applied to the limited
interference model since some orthogonality at a certain
radius is required. In \cite{Zanko_Network_2013Submeeted}, the MAC network
coding model was introduced. In the MAC network model, in contrast to the
wireless broadcast advantage based models, the superposition
property of the wireless medium is exploited. The network is
treated as a collection of multi access channels, such that each
receiver simultaneously receives data from all its in-neighbors.

\textbf{Main contributions:} This paper explores the problem of rate allocation for multicasting over slow Rayleigh fading channels using network coding. The problem is examined in a model where the network is treated as a collection of Rayleigh fading multi access channels. In our network model, we assume that links on the network vary faster than the entire network can respond to the variations. Therefore, our goal is to find a rate allocation scheme that is based solely on the statistics of the channels which minimizes the outage probability. This paper differs from prior works at two major aspects. Prior works' models assume long time averaging of the instantaneous capacity (as in the ergodic capacity approach) or averaging of the packet arrival rate (see e.g., \cite{Ho_Network_2008}). These assumptions are more suitable for fast fading model while in slow fading model this is unrealistic. Hence, in this paper we design a different rate allocation scheme which is more suitable to the slow fading model. Moreover, in this paper the design of the rate allocation scheme is based solely on the statistics which is desirable in many practical large scale networks, as will be emphasized in section \ref{sec:Rate_allocation_for_the_outage_MAC_model}.

The communication model is described in detail in section \ref{sec:Communication_model}.

In section \ref{sec:outage_probability_bounds} we present lower and upper bounds for the outage probability of a fading MAC. In section \ref{sec:Rate_allocation_for_the_outage_MAC_model} a suboptimal solution for the rate allocation problem is presented for the MAC network model. The solution is based on an upper bound on the probability of outage in the fading MAC. In section \ref{sec:distributed} a distributed solution is derived for the rate allocation problem in the MAC network model. In section \ref{sec:simulation} we report some simulation results. We end with concluding remarks.

\section{Communication model}
\label{sec:Communication_model}
Let $\mathcal{G}=\left({\mathcal{V}},{\mathcal{E}}\right)$ be a directed graph with a set of nodes $\mathcal{V}$ and directed edges $\mathcal{E}\subset\bf{V}\times\bf{V}$, where transceivers are nodes and channels are edges representing a wireless communication network. In this paper, scalars and random variables are denoted by lower case letters. Vectors and matrices are denoted by boldface lower and upper case letters, respectively. We are abusing of notation a bit by using the same letters to refer to random variables and their realizations. The cardinality of any set $\mathcal{A}$ is denoted by $|\mathcal{A}|$. All vectors are columns and inequalities between vectors are defined element-wise; i.e., ${\bf v}\leq {\bf u}$ implies $v_i\leq u_i$ for all $i$.
For any node $j\in\mathcal{V}$, we denote the in-neighborhood and out-neighborhood of $j$ by $\mathcal{I}(j)$ and $\mathcal{O}(j)$ respectively, i.e.,
$\mathcal{I}(j)=\{i:\left(i,j\right)\in\mathcal{E}\}$ and $\mathcal{O}(j)=\{i:\left(j,i\right)\in\mathcal{E}\}$. The network is treated as a collection of multi access channels, such that each receiver simultaneously receives data from all its in-neighbors. For simplicity, it is assumed that there is no interference between transmissions toward different receivers (see Fig. \ref{fig:MAC_network_model}(a)).
\begin{figure}[htbp]
	\centering
\includegraphics[width=\textwidth]{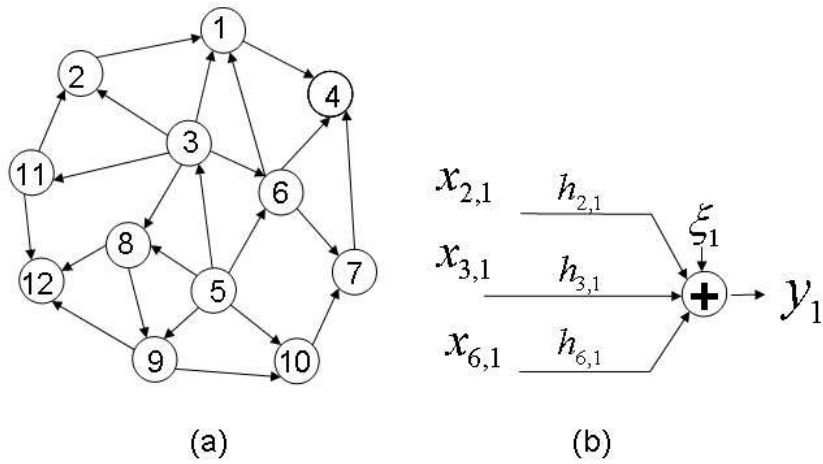}
	\caption{The MAC network model: (a) An illustration of a wireless network with $12$ transceivers positioned as a directed graph $\mathcal{G}$. In the MAC network model each receiver receives data from all its in-neighborhood nodes. For example the nodes in $\mathcal{I}(1)=\{2,3,6\}$ transmit toward node $1$ and the nodes in $\mathcal{I}(4)=\{1,6,7\}$ transmit toward $4$. However, it is assumed that (for example), there is no interference between the transmissions toward node $1$ and the transmissions toward node $4$. (b) The MAC of node $1$.}
	\label{fig:MAC_network_model}
\end{figure}
This can be achieved by orthogonal transmissions e.g., by using a certain frequency reuse pattern or directional antennas. Clearly, this is an improvement over a model where all transmissions are orthogonal. If we consider the MAC network model with deterministic channel gains, a joint power control and rate allocation solution for a (convex) network utility can be found distributively \cite{Yuan_A_Cross_2006}. This is due to the convexity of the capacity region of the multi access channels \cite{El_Gamal_Network_2011}. The deterministic model can be adapted to deal with fast fading channels in the case of a constant power allocation vector by using the ergodic MAC capacity region instead of the MAC capacity region. This ergodic capacity region is easily obtained by taking the expectations of the capacity constraints \cite{El_Gamal_Network_2011}. Here, we examine the MAC network model in the case of slow fading channels. We aim to find a rate allocation scheme that is based solely on the statistics of the channels which minimizes the outage probability.

The channel gain of link $(i,j)$ is denoted by $h_{i,j}$. $h_{i,j}$ is a zero mean circular complex normal random variable with variance of $\upsilon_{i,j}^2$. It is assumed that all $h_{i,j}$ are independent of each other. Denote by ${\bf h}_j:=\left[h_{i,j}:i\in\mathcal{I}(j)\right]$ and by $\boldsymbol{\eta_j}:=\left[|h_{i,j}|^2:i\in\mathcal{I}(j)\right]$. The transmission on link $(i,j)$ is denoted by $x_{i,j}$ and it is transmitted with an average power $p_{i,j}$. We assume that $\sigma^2_j$ is the variance of $\xi_j$ - the zero mean Gaussian noise at node $j$. Hence, the received signal at node $j$ is given by:
\begin{equation}
y_{j}=\displaystyle\sum_{i\in\mathcal{I}(j)}{h_{i,j}x_{i,j}}+\xi_j.
\end{equation}
Fig. \ref{fig:MAC_network_model}(b) illustrates the MAC of node $1$ in the network of Fig. \ref{fig:MAC_network_model}(a). The rate transmitted on a link $(i,j)$ is denoted by $r_{i,j}$, the rate allocation vector is denoted by ${\bf r}=\left[r_{i,j}:\;(i,j)\in\mathcal{E}\right]$ and the local rate allocation vector is denoted by ${\bf r}_j=\left[r_{i,j}:\;i\in\mathcal{I}(j)\right]$. When the instantaneous channel gains $h_{i,j}$ are deterministic and known, this is the well-known Gaussian multi access channel \cite{El_Gamal_Network_2011}. Hence, the instantaneous MAC capacity region is given by:
\begin{equation}
\label{ineq:MAC_capacity_Region}
\mathcal{V}^{\rm{ins}}_j({\bf h}_j):=\left\{r_{i,j}:\begin{array}{l}
\displaystyle\sum_{i\in\mathcal{M}(j)}{r_{i,j}}
\leq\log_2\left(1+\frac{P_{\mathcal{M}(j),j}}{\sigma^2_j}\right)\\
\qquad\forall\mathcal{M}(j)\subseteq\mathcal{I}(j)
\end{array}
\right\},
\end{equation}
where $P_{\mathcal{M}(j),j}=\displaystyle\sum_{i\in\mathcal{M}(j)}{p_{i,j}|h_{i,j}|^2}$.
However, when dealing with Rayleigh channels, this capacity region may not be a good measure of performance and the outage capacity is a better and more practical alternative. A common outage event is jointly declared for all links whenever we transmit toward a certain node with rates that are not supported by the instantaneous MAC capacity region.
\begin{definition}
\label{def:MAC_outage}
For a rate vector ${\bf r}_j$ and for the MAC associated with node $j$, the common outage event is
\begin{equation}
\label{eq:outage_event_def}
{\bf r}_j\notin\mathcal{V}^{\rm{ins}}_j({\bf h}_j),
\end{equation}
where the (random) capacity region $\mathcal{V}^{\rm{ins}}_j({\bf h}_j)$ is defined in (\ref{ineq:MAC_capacity_Region}).
\end{definition}
\begin{definition}
\label{def:MAC_outage_probability}
The probability of outage in the fading MAC of node $j$ is given by
\begin{equation}
P_{j}^{\rm{out}}=\Pr\left({\bf r}_j\notin\mathcal{V}^{\rm{ins}}_j({\bf h}_j)\right).
\end{equation}
\end{definition}

Similar to these definitions, we define an outage event and outage probability for
the MAC network model.
\begin{definition}
The outage event for the MAC network model is the event for which there exists node $j\in{\mathcal{V}}$ such that ${\bf r}_j\notin\mathcal{V}^{\rm{ins}}_j({\bf h}_j)$.
\end{definition}

Hence, the probability of outage in the MAC network model is given by
\begin{equation}
\label{eq:probability_outage_MAC_network_model}
P_{\rm{MAC}}^{\rm{out}}=\Pr\left(\displaystyle\bigcup_{j\in{\mathcal{V}}}{\left\{{\bf r}_j\notin\mathcal{V}^{\rm{ins}}_j({\bf h}_j)\right\}}\right).
\end{equation}

To complete the description of the local communication model, we associate the codebooks, the encoders $\left(\mathcal{F}_{i,j}\;:\;i\in\mathcal{I}(j)\right)$ and the decoder $g_j$ that establish the connection between $\mathcal{I}(j)$ and $j$ at rates $\left(r_{i,j}: i\in\mathcal{I}(j) \right)$ to any node in the network. Obviously, node $j$ shares the appropriate codebooks and encoders with its in-neighbors.

The source node is denoted by $s\in{\bf V}$ and it is assumed that $\mathcal{I}(s)=\phi$. The set of all destinations (sinks) is denoted by $\mathcal{D}_s\subseteq\bf{V}\backslash\{s\}$. Intermediate nodes are allowed to send out packets that are a combination of their received information and as a result they break the flow conservation by increasing/decreasing the outside rate. However, the main theorem of network coding for multicast is stated in terms of the max-flow (min-cut) between each source and its destinations. Therefore, we distinguish between the flow at an edge $(i,j)$ and the actual rate at that link. Let $f_{i,j}^{d}$ be the flow at edge $(i,j)$ destined for destination $d\in\mathcal{D}_s$, and let $r_{i,j}$ be the actual at edge $(i,j)$. The communication parameters are summarized in Table \ref{tab:communication_model_parameters}.

\begin{table}[htbp]
	\centering
		\caption{A summary of communication model parameters}
		\begin{tabular}{lll}
$x_{i,j}$ & Transmission on link $(i,j)$ &\\
$h_{i,j}$ & Channel gain of link $(i,j)$ & $h_{i,j}\sim\mathcal{CN}(0,\upsilon_{i,j}^2)$\\
${\bf h}_j$ & \multicolumn{2}{l}{${\bf h}_j=\left[h_{i,j}:i\in\mathcal{I}(j)\right]$}\\
${\boldsymbol{\eta}}_j$ & \multicolumn{2}{l}{${\boldsymbol{\eta}}_j=\left[|h_{i,j}|^2:i\in\mathcal{I}(j)\right]$}\\
$p_{i,j}$ & \multicolumn{2}{l}{Average power of the transmission on link $(i,j)$}\\
$\xi_j$     & Noise at node $j$ & $\xi_j\sim\mathcal{CN}(0,\sigma^2_j)$\\
$y_j$	  & Received signal at node $j$ \\
$r_{i,j}$	&	Rate at link $(i,j)$ 	\\
${\bf r}_j$	& 	Local rate allocation vector 	& ${\bf r}_j=\left[r_{i,j}:\;i\in\mathcal{I}(j)\right]$\\
${\bf r}$& Rate allocation vector &	${\bf r}=\left[r_{i,j}:\;(i,j)\in\mathcal{E}\right]$\\
$\mathcal{V}^{\rm{ins}}_j({\bf H})$ &  \multicolumn{2}{l}{Instantaneous MAC capacity region of node $j$} \\
$P_{j}^{\rm{out}}$  & \multicolumn{2}{l}{Probability of outage in the fading MAC of node $j$}\\
$P_{\rm{MAC}}^{\rm{out}}$ & \multicolumn{2}{l}{Probability of outage in the MAC network model} \\
$s$ & Source node &\\
$\mathcal{D}_s$ & Set of all destinations & $\mathcal{D}_s\subseteq\bf{V}\backslash\{s\}$\\
$f^d_{i,j}$ & \multicolumn{2}{l}{Flow at edge $(i,j)$ destined for destination $d\in\mathcal{D}_s$}
\end{tabular}
\label{tab:communication_model_parameters}
\end{table}

As was mentioned in section \ref{sec:introduction}, there is no loss of optimality by first finding the optimal rate allocation solution and then designing the coding scheme that realizes the connection. In the following section the rate allocation vector for the MAC network model is given as the solution to an optimization problem and the coding scheme that realizes the connection is assumed to be given. For large scale networks, where global network information is not available, the random network coding shown in \cite{Chou_Practical_2003,Lun_On_2008} can be employed. In general, in random network coding, intermediate nodes store all their received packets in their memory and when a packet injection occurs on an outgoing link, the node forms a packet that is a random linear combination of the packets in its memory. In order to enable decoding at the destinations, the random coefficients of the linear combinations are included in the header of the packet as side information. These coefficients are called the global encoding vector of the packet. Decoding is possible if all destinations collect enough packets with linearly-independent global encoding vectors. The algorithm shown in \cite{Lun_On_2008} for random packet level network coding was adjusted to the MAC network model in \cite{Zanko_Network_2013Submeeted}.

\section{Bounds on the probability of outage of a MAC}
\label{sec:outage_probability_bounds}
In this section we bound the outage probability of a fading MAC. To do so, we need the following notations and definitions. Consider a (slow) fading MAC with $n$ links each of which is a Rayleigh channel i.e., $h_i\sim\mathcal{CN}(0,\upsilon_i^2)\;\;i=1,2,\cdots,n$. Denote the variance of the zero mean Gaussian noise at the receiver by $\sigma^2$. For any matrix ${\bf B}$, $b_{i,j}$ denotes the entry in the $i$'th row and $j$'th column of ${\bf B}$. Let ${\bf B}^{r*}$ be a submatrix of ${\bf B}$ constructed by deleting the $r$'th row of ${\bf B}$. For $r=1$ we denote ${\bf B}^{1*}={\bf B}^{*}$. For any $n\geq 1$ let ${\bf 1}_n$,${\bf 0}_n$ be vectors with length $n$ of ones and zeros, respectively. For any $n\geq 1$, let ${\bf A}_n$ be a $(2^n-1)\times n$ matrix, such that for $n=1$ ${\bf A}_1=1$ and for $n\geq 2$
\begin{equation}
\label{eq:A_n:recursion}
{\bf A}_{n+1}=\left[\begin{array}{lll}
{\bf 0}_{2^n-1} &,& {\bf A}_{n}\\
1 &,& {\bf 0}_n^T \\
{\bf 1}_{2^n-1} &,& {\bf A}_{n}.
\end{array}
\right],
\end{equation}
i.e., each row of ${\bf A}_n$ is the binary representation of the row index (for example  ${\bf A}_2=\left[[0,1]^T,\right.$ $\left.[1,0]^T,[1,1]^T\right]^T$). For any scalar $a$ and vector ${\bf v}\in\mathds{R}^{K}$, ${\bf c}=a^{{\bf v}}-1$ is calculated point-wise; i.e., $c_i=a^{v_i}-1$.

The probability of outage of a fading MAC is given in definition \ref{def:MAC_outage_probability}. Obviously, the probability of outage can be expressed as
\begin{eqnarray}
\label{eq:MAC_outage_by_success}
\rm{Pr}^{\rm{out}}_{\rm{MAC}_n}=1-\Pr\left({\bf r}\in\mathcal{V}^{\rm{ins}}({\bf h})\right),
\end{eqnarray}
where $\bf{h}:=\left[h_1,h_2,\cdots,h_n\right]$ and $\mathcal{V}^{\rm{ins}}({\bf h})$ is the instantaneous capacity region. As can be seen from (\ref{ineq:MAC_capacity_Region}) the expression
${\bf r}\in\mathcal{V}^{\rm{ins}}({\bf h})$ stands for a conjunction of $(2^n-1)$ inequalities, each of which is in the form of
\begin{equation}
\label{ineq:MAC_success_inequality}
\displaystyle\sum_{i\in\mathcal{M}}{r_{i}}
\leq\log_2\left(1+\frac{P_{\mathcal{M}}}{\sigma^2}\right),
\end{equation}
where $P_{\mathcal{M}}=\displaystyle\sum_{i\in\mathcal{M}}{p_{i}|h_{i}|^2}$ and $\mathcal{M}$ is a subset of $\{1,2,\ldots,n\}$.
Rewriting (\ref{ineq:MAC_success_inequality}) in a matrix form yields
\begin{equation}
\label{ineq:MAC_success_inequality_matrix_form}
{\bf a}_{\mathcal{M}}^T {\bf r}\leq \log_2\left(1 + {\bf a}_{\mathcal{M}}^T\frac{1}{\sigma^2}{\bf P}{\boldsymbol{\eta}}\right),
\end{equation}
where, ${\bf a}_{\mathcal{M}}$ is a vector with length $n$ such that
\begin{equation}
a_i=\begin{cases}
		1 & i\in\mathcal{M}\\
		0 & \rm{otherwise},
	\end{cases}
\nonumber
\end{equation}
${\bf P}$ is an $n\times n$ diagonal matrix with $p_1,p_2,\ldots,p_n$ on the main diagonal and $\boldsymbol{\eta}=\left[|h_1|^2,|h_2|^2,\ldots,|h_n|^2\right]^T$. A simple algebraic operation yields that (\ref{ineq:MAC_success_inequality_matrix_form}) is equivalent to
\begin{equation}
\label{ineq:MAC_success_inequality_exp_matrix_form}
2^{{\bf a}_{\mathcal{M}}^T {\bf r}}-1 \leq  {\bf a}_{\mathcal{M}}^T\frac{1}{\sigma^2}{\bf P}\boldsymbol{\eta}.
\end{equation}
Note that since $|h_i|^2\;\;i=1,2,\ldots,n$ are independent exponential random variables with an expectation of $2\upsilon^2_i$, the random variables ${z}_i=\frac{1}{2\upsilon^2_i} |h_i|^2\;\;i=1,2,\ldots,n$ are i.i.d exponential random variables with an expectation of $1$. Hence, the event in (\ref{ineq:MAC_success_inequality_exp_matrix_form}) is equivalent to the event
\begin{equation}
\label{ineq:MAC_success_inequality_standard_exp_matrix_form}
2^{{\bf a}_{\mathcal{M}}^T {\bf r}}-1 \leq  {\bf a}_{\mathcal{M}}^T\frac{1}{\sigma^2}{\bf P}{\boldsymbol \Upsilon}{\bf z}_n,
\end{equation}
where ${\boldsymbol \Upsilon}$ is a diagonal matrix with $2\upsilon^2_1,2\upsilon^2_2,\ldots,2\upsilon^2_n$ on the main diagonal and ${\bf z}_n=[z_1,z_2,\ldots,z_n]^T$ is a vector of $n$ i.i.d standard exponential random variables; $E\{z_i\}=1$.
Therefore, from (\ref{eq:MAC_outage_by_success}) and (\ref{ineq:MAC_success_inequality_standard_exp_matrix_form}), it is implied that the outage probability in a MAC with $n$ links can be written as
\begin{equation}
\label{eq:Probability_of_success:A_N}
\rm{Pr}^{\rm{out}}_{\rm{MAC}_n}=1-
\Pr\left({\bf A}_n{\bf D}_n{\bf z}_n \geq {\bf b}_n\right),
\end{equation}
where ${\bf D}_n$ is a diagonal matrix with $\frac{1}{\lambda_1},\frac{1}{\lambda_2},\cdots,\frac{1}{\lambda_n}$ on the main diagonal, $\lambda_i=\frac{1}{2\upsilon_i^2}\frac{\sigma^2}{p_i}$ and ${\bf b}_n=2^{{\bf A}_n{\bf r}}-1$. For example in a MAC with $n=3$ links we have that
\begin{equation}
\begin{array}{ll}
	{\bf A}_3=\left[\begin{array}{lll}
	0&0&1\\0&1&0\\0&1&1\\1&0&0\\1&0&1\\1&1&0\\1&1&1
	\end{array}\right],
& {\bf b}_3=\left[\begin{array}{l}
2^{r_3}-1\\2^{r_2}-1\\2^{r_2+r_3}-1\\2^{r_1}-1\\2^{r_1+r_3}-1\\2^{r_1+r_2}-1\\2^{r_1+r_2+r_3}-1
\end{array}
\right].
\end{array}
\nonumber
\end{equation}
Note that when the MAC is with i.i.d links we have that ${\bf D}_n=\frac{1}{\lambda}{\bf I}_n$ where ${\bf I}_n$ is an $n\times n$ identity matrix.

Hence, the probability of outage in a MAC is related to the joint distribution of linear combinations of exponential random variables. Huffer and Lin \cite{Huffer_Computing_2001} presented an algorithm for the computation of the exact expression of the joint distribution of general linear combinations of spacings\footnote{Suppose $u_i\;\;i=1,2,\cdots,n$ are independently and uniformly distributed on the interval (0,1), and let $u_{(1)}\leq u_{(2)}\leq\cdots\leq u_{(n)}$ be the corresponding order statistics. The spacings $s_1,s_2,\cdots,s_{n+1}$ are defined by the successive differences between
the order statistics: $s_i=u_{(i)}-u_{(i-1)}$, where $u_{(0)}:=0$ and $u_{(n+1)}:=1$.} by  repeated uses of two recursions that reduce the dimensionality of the problem. They also pointed out that the algorithm remains valid as well for linear combinations of exponential random variables. However, this is inaccurate and in this paper we revise result to handle exponential random variables. The new recursion is given in Lemma \ref{lemma:Huffer_computing_2001} in Appendix \ref{appendix:recursion}.

By using the algorithm in \cite{Huffer_Computing_2001} together with Lemma \ref{lemma:Huffer_computing_2001}, an exact expression of the probability of common outage can be computed. However, the computation of a symbolic expression becomes extremely complicated in a MAC with more than $2$ links. Therefore, we present an upper and a lower bound on that outage probability. To that end, we need the following lemma:
\begin{lemma}
\label{lemma:A_n_has_a_row_with_one_and_zeros}
Let $\bf z$ be a vector of $n$ i.i.d exponential random variables with an expectation $E\{z_i\}=1$. If there exists an entry $a_{r,i}>0$ in $\bf A$ such that $a_{r,j}=0$ for all $j\neq i$ and $b_r\geq 0$, the following holds
\begin{equation}
\Pr\left({\bf Az}>\lambda{\bf b}\right)=e^{-\lambda\frac{b_r}{a_{r,i}}}
\Pr\left({\bf A}^{r*}{\bf z}>\lambda\left({\bf b}^{r*}-\frac{b_r}{a_{r,i}}{\bf a}^{r*}\right)\right).
\nonumber
\end{equation}
\end{lemma}
\begin{proof}
{This lemma is an immediate consequence of Lemma \ref{lemma:Huffer_computing_2001} (i.e., when $k=1$ and $\Psi=\{j\}$).
}
\end{proof}

In \cite{Rao_Outage_2013} it was pointed out that the outage probability of a MAC with $n$ i.i.d links is bounded from below by
\begin{equation}
\label{ineq:Lower_MAC_Outage_Bound_gamma}
\rm{Pr}^{\rm{out}}_{\rm{MAC}_n}\geq 1-e^{-\lambda S_n}\frac{\Gamma\left(n,\lambda\left(2^{R_n}-1-S_n\right)\right)}{(n-1)!},
\end{equation}
where $\Gamma(n,x)$ is the incomplete gamma function, $S_n=\displaystyle\sum_{k=1}^{n}{(2^{r_n}-1)}$ and $R_n=\displaystyle\sum_{i=1}^{n}{r_i}$.
In the case of a Rayleigh fading MAC with independent links but with different variances, Theorem \ref{theorem:lower_bound_outage_probability} gives a lower bound on the probability of outage.
\begin{theorem}
\label{theorem:lower_bound_outage_probability}
Let $\lambda_i=\frac{1}{2\nu^2}\frac{\sigma}{p_i}$ $i=1,2,\cdots,n$ have distinct values; i.e., $\lambda_i\neq \lambda_j$ for all $i\neq j$, then the probability of a MAC with $n$ independent $\rm{Rayleigh}(\upsilon_i)$ channel $i=1,2,\cdots,n$ is bounded below by
\begin{equation}
\rm{Pr}^{\rm{out}}_{\rm{MAC}_n}\geq 1-
\displaystyle\sum_{i=1}^{n}{\gamma_i e^{-\beta_n-\lambda_i\left(2^{R_n}-S_n-1\right)}},
\end{equation}
where,
$\gamma_i=\displaystyle\prod_{j\neq i}{\frac{\lambda_j}{\lambda_j-\lambda_i}}$, $\beta_n=\displaystyle\sum_{i=1}^{n}{\lambda_i(2^{r_i}-1)}$, $S_n=\displaystyle\sum_{i=1}^{n}{(2^{r_i}-1)}$ and $R_n=\displaystyle\sum_{i=1}^{n}{r_i}$.
\end{theorem}
\begin{proof}
As was explained earlier, the expression ${\bf A}_n{\bf D}_n{\bf z}_n \geq {\bf b}_n$ in equation (\ref{eq:Probability_of_success:A_N}) stands for a conjunction of $(2^n-1)$ inequalities. The inequalities that are related to the rows of ${\bf A}_n$ indexed by $2^i\;\;i=0,1,\ldots,(n-1)$ stand for the direct instantaneous capacity constraints $r_i\leq\log_2\left(1+\frac{|h_i|^2 p_i}{\sigma^2}\right)$, whereas all the other inequalities refer to constraints that involve more than one link of the MAC capacity region (see eq. (\ref{ineq:MAC_capacity_Region})). In particular, the $(2^n-1)^{\rm{th}}$ inequality refers to the constraint
\begin{equation}
\displaystyle\sum_{i=1}^n{r_i}\leq\log_2\left(1+\displaystyle\sum_{i=1}^n\frac{p_i|h_i|^2}{\sigma^2}\right).
\nonumber
\end{equation}
Obviously,
\begin{equation}
\label{ineq:submatrix_of_A_n}
\Pr\left({\bf A}_n{\bf D}_n{\bf z}_n \geq {\bf b}_n\right)\leq\Pr\left(\tilde{\bf A}_n{\bf D}_n{\bf z}_n \geq \tilde{\bf b}_n\right),
\end{equation}
where $\tilde{\bf A}_n$ is a submatrix of ${\bf A}_n$ constructed by taking the rows indexed by $\{2^i\;:\;i=0,1,\ldots,(n-1)\}\cup\{2^n-1\} $ of ${\bf A}_n$ and $\tilde{\bf b}_n$ is a sub-vector of ${\bf b}_n$ constructed by taking the appropriate entries of ${\bf b}_n$.
Note that up to a permutation of the rows, the first $n$ rows of $\tilde{\bf A}_n$ is the identity matrix. Therefore, we can eliminate these $n$ rows by $n$ uses of lemma \ref{lemma:A_n_has_a_row_with_one_and_zeros}. Hence, we have that
\begin{equation}
\Pr\left(\tilde{\bf A}_n{\bf D}_n{\bf z}_n \geq \tilde{\bf b}_n\right)=e^{-\beta_n}\Pr\left(\displaystyle\sum_{i=1}^n{\frac{z_i}{\lambda_i}}\geq x\right),
\end{equation}
where $x=2^{R_n}-S_n-1$. The probability of a distinct coefficients linear combination of i.i.d exponential variables is given
by \cite{Huffer_Divided_1988}
\begin{equation}
\label{eq:distinct_iid_exponential_CDF}
\Pr\left(\displaystyle\sum_{i=1}^n{\frac{z_i}{\lambda_i}}\geq x\right)=
\displaystyle\sum_{i=1}^{n}{\gamma_i e^{-\lambda_i x}}.
\end{equation}
The claim now follows.
\end{proof}

Note that Theorem \ref{theorem:lower_bound_outage_probability} is valid only when $\lambda_1,\lambda_2,\ldots,\lambda_n$ are all distinct. When we have a MAC with a set of $K$ links with the same value of $\lambda_i$ and $n-K$ links with distinct values of $\lambda$, a similar bound can be computed by replacing the probability in (\ref{eq:distinct_iid_exponential_CDF}) by integrating the expression of the pdf derived in \cite{Khuong_General_2006}. We omit the calculation of this probability here, for the sake of brevity.

For simplicity, in the derivation of the upper bound we only consider a MAC with $n$ i.i.d links. Computing the bound for the case where there are independent links with nonidentical variances is much complicated and is considered as future work. The upper bound for a MAC with $3$ links is given in Lemma \ref{lemma:MAC_3_bound} and the upper bound for the general case is given in Theorem \ref{theorem:MAC_probability_lower_bound}.
\begin{lemma}
\label{lemma:MAC_3_bound}
The probability of common outage of a MAC with $3$ i.i.d links is bounded by
\begin{equation}
\rm{Pr}^{\rm{out}}_{\rm{MAC}_3}\leq 1-e^{-\lambda\left(2^{R_3}-1\right)}G(\lambda\alpha_3)
\end{equation}
where $G(x)=\frac{1}{2}x^2+x+1$, $\lambda=\frac{1}{2\upsilon^2}\frac{\sigma^2}{p}$, $R_3=\displaystyle\sum_{i=1}^{3}{r_i}$ and $\alpha_3=(2^{r_1}-1)(2^{r_2}-1)(2^{r_3}-1)$.
\end{lemma}
\begin{proof}
The probability of a successful (non-outage) transmission is given by
\begin{equation}
1-\rm{Pr}^{\rm{out}}_{\rm{MAC}_3}=\Pr\left({\bf A}_3{\bf z}_3\geq \lambda {\bf b}_3\right).
\end{equation}
Define the following constants $\beta_i=2^{r_i}-1$ and $\beta_{i,j}=2^{r_i+r_j}-1$. Note that the rows of ${\bf A}_3$ and ${\bf b}_3$ indexed by $\{2^i\;:\;i=0,1,2\}$ satisfy the conditions of lemma \ref{lemma:A_n_has_a_row_with_one_and_zeros}. Hence, by three uses of Lemma \ref{lemma:A_n_has_a_row_with_one_and_zeros} we can eliminate these rows of ${\bf A}$ and ${\bf b}_3$. These three uses of Lemma \ref{lemma:A_n_has_a_row_with_one_and_zeros} are legitimate, since after each use of the lemma the result matrix and vector still satisfy the conditions of Lemma \ref{lemma:A_n_has_a_row_with_one_and_zeros} (since $\beta_i\geq 0$, $\beta_{i,j}-\displaystyle\sum_{i\in{\mathcal{B}}}{\beta_i}\geq 0$ and $(2^{R_3}-1)-\displaystyle\sum_{i=1}^{j}{\beta_i}\geq 0$, for all $\mathcal{B}\subseteq\{i,j\}$ and $i,j\in\{1,2,3\}$). Therefore, the probability of successful transmission can be rewritten as
\begin{equation}
\label{eq:probability_of_success_A(4)}
1-\rm{Pr}^{\rm{out}}_{\rm{MAC}_3}=e^{-\lambda S_3}\Pr\left({\bf A}{\bf z}\geq\lambda\tilde{\bf b}\right),
\end{equation}
where ${\bf A}$ is a submatrix of ${\bf A}_3$ constructed by deleting the rows indexed by $\{2^i:i=0,1,2\}$, $S_3=\displaystyle\sum_{i=1}^{3}{(2^{r_i}-1)}$ and
$\tilde{\bf b}=\left[\beta_{2,3}-\beta_2-\beta_3\right.$
,$\beta_{1,3}-\beta_1-\beta_3$
,$\beta_{1,2}-\beta_1-\beta_2$
,$\left.(2^{R_3}-1)-\beta_1-\beta_2-\beta_3\right]^T$.
It is easy to see that
\begin{equation}
\label{eq:A(4)_inequality}
\Pr\left({\bf A}{\bf z}\geq\lambda\tilde{\bf b}\right)\geq\Pr\left(\tilde{\bf A}{\bf z}\geq\lambda\tilde{\bf b}\right),
\end{equation}
where
\begin{equation}
\begin{array}{ll}
\bf A=\left[\begin{array}{lll}
	0&1&1\\1&0&1\\1&1&0\\1&1&1
	\end{array}\right],
& \tilde{\bf A}=\left[\begin{array}{lll}
	0&1&0\\0&0&1\\1&0&0\\1&1&1
	\end{array}\right],
\end{array}
\nonumber
\end{equation}
since the eliminated $z_1,z_2$ and $z_3$ are non-negative random variables.
Note that since  $\tilde{b}_4-\displaystyle\sum_{i=0}^{3}{\tilde{b}_i}=\alpha\geq 0$ and $\tilde{\bf b}\geq 0$, we have that $\tilde{b}_4-\displaystyle\sum_{i=0}^{j}{\tilde{b}_i}\geq 0$ for all $j\in\{1,2,3\}$. Again, by three uses of Lemma \ref{lemma:A_n_has_a_row_with_one_and_zeros}, we can eliminate the first three rows of $\tilde{\bf A}$ and write
\begin{equation}
\label{eq:A_tilde_probability}
\Pr\left(\tilde{\bf A}{\bf z}\geq\lambda\tilde{\bf b}\right)=e^{-\lambda\tilde{\gamma}}\Pr\left(z_1+z_2+z_3\geq\lambda\alpha_3\right),
\end{equation}
where $\tilde{\gamma}=\tilde{b}_1+\tilde{b}_2+\tilde{b}_3$.
Note that $Z=z_1+z_2+z_3$ has an $\rm{Erlang}(3,1)$ distribution and therefore
\begin{equation}
\Pr\left(Z>z\right)=e^{-z}G(z).
\end{equation}
Hence, (\ref{eq:A_tilde_probability}) can be rewritten as
\begin{equation}
\label{eq:A_tilde_probability_rewritten}
\Pr\left(\tilde{\bf A}{\bf z}\geq\lambda\tilde{\bf b}\right)=e^{-\lambda\tilde{\gamma}}e^{-\lambda\alpha_3}G(\lambda\alpha_3).
\end{equation}
Combining (\ref{eq:probability_of_success_A(4)}),(\ref{eq:A(4)_inequality}) and (\ref{eq:A_tilde_probability_rewritten}) yields
\begin{equation}
1-\rm{Pr}^{\rm{out}}_{\rm{MAC}_3}\geq e^{-\lambda\left(\alpha_3+\gamma+\tilde{\gamma}\right)}G(\lambda\alpha_3).
\end{equation}
The claim now follows from the fact that
\begin{equation}
\alpha_3+\gamma+\tilde{\gamma}=2^{R_3}-1.
\nonumber
\end{equation}
\end{proof}

\begin{theorem}
\label{theorem:MAC_probability_lower_bound}
The probability of common outage of a MAC with $n\geq 3$ i.i.d $\rm{Rayleigh}(\upsilon)$ channels is bounded by
\begin{equation}
\label{ineq:Upper_MAC_Outage_Bound}
\rm{Pr}^{\rm{out}}_{\rm{MAC}_n}\leq 1-e^{-\lambda\left(2^{R_n}-1\right)}\tilde{G}(\lambda\alpha_n)
\end{equation}
where $G(x)=\frac{1}{2}x^2+x+1$, $R_n=\displaystyle\sum_{i=1}^{n}{r_i}$ and $\alpha_n=\displaystyle\prod_{i=1}^{n}{(2^{r_i}-1)}$.
\end{theorem}
\begin{proof}
The proof proceeds by induction on $n$. The claim is true for $n=3$ by Lemma \ref{lemma:MAC_3_bound}. Let the statement be true for $n=k$. We will now
prove the result for $n=k+1$. From (\ref{eq:Probability_of_success:A_N}) we have that
\begin{equation}
\label{eq:Probability_of_success:A_k+1}
1-\rm{Pr}^{\rm{out}}_{\rm{MAC}_{k+1}}=
\Pr\left({\bf A}_{k+1}{\bf z}_{k+1}\geq \lambda {\bf b}_{k+1}\right).
\end{equation}
Note that $\left[{\bf b}_{k+1}\right]_{2^k}=2^{r_1}-1$. Also note that by exploiting the structure of ${\bf A}_{k+1}$ (see equation (\ref{eq:A_n:recursion})), the $2^k$'th row of ${\bf A}_{k+1}$ can be eliminated by using Lemma \ref{lemma:A_n_has_a_row_with_one_and_zeros}. Hence,
\begin{equation}
\label{eq:Probability_of_success:A_k+1,reduced}
1-\rm{Pr}^{\rm{out}}_{\rm{MAC}_{k+1}}=
e^{-\lambda(2^{r_1}-1)}\Pr\left(\tilde{\bf A}_{k+1}{\bf z}_{k+1}\geq \lambda \tilde{\bf b}_{k+1}\right),
\end{equation}
where
\begin{equation}
\label{eq:A_k+1:reduced}
\tilde{\bf A}_{k+1}=\left[\begin{array}{lll}
{\bf 0}_{2^k-1} &,& {\bf A}_{k}\\
{\bf 1}_{2^k-1} &,& {\bf A}_{k}
\end{array}
\right]
\end{equation}
and $\tilde{\bf b}_{k+1}=\left(2^{\tilde{\bf A}_{k+1}{\bf r}}-1\right)-(2^{r_1}-1)\left[{\bf 0}_{2^k-1}^T, {\bf 1}_{2^k-1}^T\right]^T$.
It can easily be verified that
\begin{equation}
\label{eq:Probability_A_check}
\Pr\left(\tilde{\bf A}_{k+1}{\bf Z}_{k+1}\geq \lambda \tilde{\bf b}_{k+1}\right)\geq
\Pr\left(\check{\bf A}_{k+1}{\bf Z}_{k+1}\geq \lambda \tilde{\bf b}_{k+1}\right),
\end{equation}
where
\begin{equation}
\check{\bf A}_{k+1}=\left[\begin{array}{lll}
{\bf 0}_{2^k-1} &,& {\bf A}_{k}\\
{\bf 0}_{2^k-1} &,& {\bf A}_{k}
\end{array}
\right],
\nonumber
\end{equation}
since $z_1\geq 0$.
Note that for any two vectors ${\bf 0}_{2^k-1}\leq{\bf x}_1\leq{\bf x}_2$ the following holds
\begin{equation}
\label{eq:A_check_reduced}
\Pr\left(\check{\bf A}_{k+1}{\bf z}_{k+1}\geq \left[{\bf x}_1^T,{\bf x}_2^T\right]^T\right)=\Pr\left({\bf A}_k{\bf z}_k\geq {\bf x}_2\right).
\end{equation}
Note that
\begin{equation}
2^{\tilde{\bf A}_{k+1}{\bf r}}=\left[\begin{array}{l}
2^{{\bf A}_k\boldsymbol{\gamma}}\\
2^{r_1}\cdot 2^{{\bf A}_k\boldsymbol{\gamma}}
\end{array}
\right],
\end{equation}
where $\boldsymbol{\gamma}=[r_2,r_3,\ldots,r_{k+1}]^T$. Therefore, we have that
\begin{equation}
\label{eq:b_k+1_tilde}
\tilde{\bf b}_{k+1}=\left[\begin{array}{l}
2^{{\bf A}_k\boldsymbol{\gamma}}-1\\
2^{r_1}\left(2^{{\bf A}_k\boldsymbol{\gamma}}-1\right)
\end{array}
\right].
\end{equation}
Since $2^{r_1}\geq 1$, combining
(\ref{eq:A_check_reduced}) with (\ref{eq:b_k+1_tilde}) yields
\begin{equation}
\label{eq:A_check_reduced_even_more}
\Pr\left(\check{\bf A}_{k+1}{\bf z}_{k+1}\geq \lambda\tilde{\bf b}_{k+1}\right)=\Pr\left({\bf A}_k{\bf z}_k\geq \tilde{\lambda}\left(2^{{\bf A}_k\boldsymbol{\gamma}}-1\right)\right),
\end{equation}
where $\tilde{\lambda}=\lambda 2^{r_1}$. By the induction
hypothesis for $k$, we have that
\begin{equation}
\label{eq:A_k:induction_hypothesis_lambda_tilde}
\Pr\left({\bf A}_k{\bf z}_k\geq \tilde{\lambda}\left(2^{{\bf A}_k\boldsymbol{\gamma}}-1\right)\right)\geq
e^{-\tilde{\lambda}\left(2^{\tilde{R}_2}-1\right)}G\left(\tilde{\lambda}\tilde{\alpha}_2\right),
\end{equation}
where $\tilde{R}_2=\displaystyle\sum_{i=2}^{k+1}{r_i}$ and $\tilde{\alpha}_2=\displaystyle\prod_{i=2}^{k+1}{\left(2^{r_i}-1\right)}$.
Combining (\ref{eq:Probability_of_success:A_k+1,reduced}),(\ref{eq:Probability_A_check}),(\ref{eq:A_check_reduced_even_more}) and (\ref{eq:A_k:induction_hypothesis_lambda_tilde}) yields
\begin{equation}
1-\rm{Pr}^{\rm{out}}_{\rm{MAC}_{k+1}}\geq
e^{-\lambda(2^{r_1}-1)}e^{-\tilde{\lambda}\left(2^{\tilde{R}_2}-1\right)}G\left(\tilde{\lambda}\tilde{\alpha}_2\right).
\end{equation}
The claim now follows from the fact that
\begin{equation}
\lambda(2^{r_1}-1)+\tilde{\lambda}\left(2^{\tilde{R}_2}-1\right)=\lambda\left(2^{R_{k+1}}-1\right),
\nonumber
\end{equation}
$\tilde{\lambda}\tilde{\alpha}_2\geq \lambda\alpha_{k+1}$ and that $G(x)$ is monotonically increased with $x$ .
\end{proof}

Note that the lower bound (\ref{ineq:Lower_MAC_Outage_Bound_gamma}) in the i.i.d case may also be expressed as
\begin{equation}
\label{ineq:Lower_MAC_Outage_Bound}
\rm{Pr}^{\rm{out}}_{\rm{MAC}_n}\geq 1-e^{-\lambda\left(2^{R_n}-1\right)}\tilde{G}\left(\lambda\left(2^{R_n}-1-S_n\right)\right),
\end{equation}
where,$\tilde{G}\left(x\right)=\displaystyle\sum_{k=0}^{n-1}{\frac{1}{k!}x^k}$.
Hence,
\begin{equation}
\tilde{G}(\lambda\beta_n)\geq e^{\lambda\left(2^{R_n}-1\right)}\left(1-\rm{Pr}^{\rm{out}}_{\rm{MAC}_n}\right)\geq G(\lambda\alpha_n),
\end{equation}
where $\beta_n=2^{R_n}-1-S_n$.
\section{Rate allocation for the fading MAC network model}
\label{sec:Rate_allocation_for_the_outage_MAC_model}
In this section we study the problem of finding the rate allocation vector for the fading MAC network model discussed in the previous sections. In our wireless model we assume a slow fading model with independent Rayleigh fading channels. For simplicity we only consider a network in which $\lambda_{i,j}=\lambda_j$ for all $j\in\mathcal{V}\backslash\{s\}$. In other words, the network is assumed to be a collection of multiple access channels with i.i.d links (note that $\lambda_j\;\;j\in\mathcal{V}\backslash\{s\}$ may be distinct). This assumption comes down to the fact that we normalized the transmission power of nodes that connected to the same receiver appropriately to the statistics of the best channel. As was mentioned earlier, the case where there are independent links with nonidentical variances is much complicated and is considered as future work.

While it is assumed that the instantaneous channels gain may be available at both the encoders and the decoders, we assume that the rate allocation vector is determined a-priori, based solely on the statistics. The rationale for this assumption is that the rate allocation vector is determined based on network considerations, whereas the instantaneous state of each component of the network varies faster than the entire network can respond to the variations. Note that this assumption is practical as well when the power constraints must be satisfied in each encoding block (e.g. when we are under Federal Communications Commission (FCC) regulations).

Consider the rate allocation graph $\tilde{\mathcal{G}}=\left(\mathcal{V},\mathcal{E},{\bf r}\right)$ (the graph that is obtained by assigning a rate $r_{i,j}$ for each link $(i,j)$ in $\mathcal{G}$). The optimal rate allocation vector that minimizes the probability of outage in the MAC network model while maintaining a multicast rate of $R_s$ is a solution of the following optimization problem:
\begin{subequations}
\label{opt:optimal_MAC_convex}
\begin{align}
\tag{\ref{opt:optimal_MAC_convex}}
&\displaystyle\min_{{\bf f},{\bf r}}
{\;\;\Pr\left(\displaystyle\bigcup_{j\in{\mathcal{V}}}{\left\{{\bf r}_j\notin\mathcal{V}^{\rm{ins}}_j({\bf h}_j)\right\}}\right)}\\
&\qquad\rm{subject\;\;to}\nonumber\\
\label{con:flow_rate_positivity_MAC_convex_optimal}
&0\leq f_{i,j}^{d} \leq r_{i,j} \;\;\;\forall(i,j)\in\mathcal{E},d\in\mathcal{D}_s\\
\label{con:flow_conservation_MAC_convex_optimal}
&\displaystyle\sum_{i\in\mathcal{I}(j)}{f_{i,j}^{d}}-
\displaystyle\sum_{i\in\mathcal{O}(j)}{f_{j,i}^{d}}=\begin{cases}
												0 & j\notin \{s,d\}\\
												R_s & j=d\\
											\end{cases}\\
&\qquad\qquad\qquad\qquad\qquad\forall j\in\mathcal{V}\backslash \{s\},d\in\mathcal{D}_s,\nonumber
\end{align}
\end{subequations}
where the flow constraints (\ref{con:flow_rate_positivity_MAC_convex_optimal})-(\ref{con:flow_conservation_MAC_convex_optimal}) guarantee that any feasible solution of (\ref{opt:optimal_MAC_convex}) provides a minimum min-cut of at least $R_s$ between the source and each destination. Therefore, a multicast rate of $R_s$ is achievable by network coding, see Theorem 1. in \cite{Lun_Minimum_2006}.

Unfortunately, as was mentioned earlier, the computation of the probability of a common outage becomes extremely complicated in a MAC with more than $2$ links. Therefore, we present a suboptimal solution to the rate allocation problem in the fading MAC network model. We relax the problem and instead of using the exact expression of the probability of common outage, we minimize an upper bound on the outage probability of a multiple access channel. To that end, consider the following lemma.
\begin{lemma}
\label{lemma:weaker_bound}
The probability of common outage of a MAC with $n$ i.i.d $\rm{Rayleigh}(\upsilon)$ channels is bounded by
$\rm{Pr}^{\rm{out}}_{\rm{MAC}_n}\leq 1-e^{-\lambda\left(2^{R_n}-1\right)}$.
\end{lemma}
\begin{proof}
For $n=1$, we have a Rayleigh fading Gaussian channel with outage probability given by:
\begin{equation}
\rm{Pr}^{\rm{out}}_{\rm{MAC}_1}=1-e^{-\lambda \left(2^{r_1}-1\right)}.
\end{equation}
For $n=2$, a simple computation yields that
\begin{align}
1-\rm{Pr}^{\rm{out}}_{\rm{MAC}_2}&=e^{-\lambda\left(2^{r_1+r_2}-1\right)}\left(1+\lambda\left(2^{r_1}-1\right)\left(2^{r_2}-1\right)\right)\nonumber\\
&\geq e^{-\lambda\left(2^{r_1+r_2}-1\right)}.
\end{align}
For $n\geq 3$, the claim follows from Theorem \ref{theorem:MAC_probability_lower_bound} and the fact that for any non-negative $x$, we have $G(x)\geq 1$.
\end{proof}

The probability of an outage in the MAC network model is given in equation (\ref{eq:probability_outage_MAC_network_model}). By assumption, all $h_{i,j}$ are independent of each other. Therefore, the probability of an outage in the MAC network model can be rewritten as
\begin{equation}
\label{eq:probability_outage_MAC_network_model_simplified}
P_{\rm{MAC}}^{\rm{out}}=1-\displaystyle\prod_{j\in\mathcal{V}\backslash\{s\}}{\left(1-P_j^{\rm{out}}\right)}.
\end{equation}
Obviously, if $P_j^{\rm{out}}$ is bounded above by $\tilde{P}_j^{\rm{out}}$ the following holds
\begin{equation}
P_{\rm{MAC}}^{\rm{out}}\leq 1-\displaystyle\prod_{j\in\mathcal{V}\backslash\{s\}}{\left(1-\tilde{P}_j^{\rm{out}}\right)}.
\end{equation}
Although the bound in Lemma \ref{lemma:weaker_bound} is weaker than the one we get from Theorem \ref{theorem:MAC_probability_lower_bound}, we used the weaker bound to find a rate allocation vector for the outage MAC model. For every $j\in\mathcal{V}\backslash\{s\}$ denote
\begin{equation}
\label{eq:r_j_tilde}
\tilde{R}_j:=\displaystyle\sum_{i\in\mathcal{I}(j)}{r_{i,j}}.
\end{equation}
Hence, from Lemma \ref{lemma:weaker_bound} we have
\begin{equation}
P_{\rm{MAC}}^{\rm{out}}\leq 1-\displaystyle\prod_{j\in\mathcal{V}\backslash\{s\}}{e^{-\lambda_j\left(2^{\tilde{R}_j}-1\right)}}.
\end{equation}
Finally, since $e^{-\lambda\left(2^{R_j}-1\right)}$ is log-concave the problem becomes computationally tractable:
\begin{subequations}
\label{opt:suboptimal_MAC_convex}
\begin{align}
\tag{\ref{opt:suboptimal_MAC_convex}}
&\displaystyle\min_{{\bf f},{\bf r}}
{\displaystyle\sum_{j\in\mathcal{V}\backslash \{s\}}
{\lambda_j 2^{\tilde{R}_j}}}\\
&\qquad\rm{subject\;\;to}\nonumber\\
\label{con:flow_rate_positivity_MAC_convex}
&0\leq f_{i,j}^{d} \leq r_{i,j} \;\;\;\forall(i,j)\in\mathcal{E},d\in\mathcal{D}_s\\
\label{con:flow_conservation_MAC_convex}
&\displaystyle\sum_{i\in\mathcal{I}(j)}{f_{i,j}^{d}}-
\displaystyle\sum_{i\in\mathcal{O}(j)}{f_{j,i}^{d}}=\begin{cases}
												0 & j\notin \{s,d\}\\
												R_s & j=d\\
											\end{cases}\\
&\qquad\qquad\qquad\qquad\qquad\forall j\in\mathcal{V}\backslash \{s\},d\in\mathcal{D}_s.\nonumber
\end{align}
\end{subequations}

\section{Distributed solution for MAC network model}
\label{sec:distributed}
In the previous section the rate allocation vector for the MAC network model was given as a solution to the convex optimization problem (\ref{opt:suboptimal_MAC_convex}). This problem can easily be solved by a standard convex optimization technique in a centralized fashion. However, the centralized solution requires full knowledge of the network topology and statistics. In this section we discuss how to distributively solve (\ref{opt:suboptimal_MAC_convex}).

As pointed out in e.g., \citep{Lun_Minimum_2006,Strikant_the_mathematics_2004,Feijer_stability_2010}, if certain conditions are satisfied, convex optimization problems may be solved distributively by a continuous time primal-dual method. This method can be described as follows. The optimization is studied through its Lagrangian where the primal and dual variables are updated simultaneously by a set of gradient laws (dynamic system). These laws define a trajectory in the direction of the respective partial gradients, starting from some initial point. The dynamic system is stated such that the saddle points of the Lagrangian are equilibrium points. Hence, if a strong duality holds for the original convex optimization problem, the algorithm stops updating the variables when it reaches the optimal solution. It is worth mentioning that in contrast to gradient method, in which convergence is guaranteed for convex problems from any initial point (see e.g., \S 9 in \cite{Boyd_Convex_2004}), the asymptotic behavior of dynamic systems is not immediate in the general case (even though the problem is convex). In other words, convergence to an equilibrium point is not guaranteed in the general case. In this type of problem the existence of Lyapunov functions is used to prove the stability of the equilibrium points. There is no general technique for the construction of these functions. However, in some specific cases the construction of Lyapunov functions is known (see e.g.,  \cite{Khalil_Nonlinear_2002}).

As can be easily verified, the cost function in (\ref{opt:suboptimal_MAC_convex}) is not strictly convex (and also is not separable in the decision variables ${\bf f}$ and ${\bf r}$). In problems with a non-strictly convex cost function, it is possible to have more than one optimum point. Hence, in this case the standard primal-dual solution may not converge. In \cite{Feijer_stability_2010} a modified primal dual gradient method was derived for non-strictly convex problems. In that method the solution will converge to one of the optimal points by modifying the constraint set of the convex optimization problem. Following \cite{Feijer_stability_2010}, we suggest the following modified convex optimization problem:
\begin{subequations}
\label{opt:MAC_equivalent}
\begin{align}
\tag{\ref{opt:MAC_equivalent}}
&\displaystyle\min_{{\bf f},{\bf r}}
{\displaystyle\sum_{j\in\mathcal{V}\backslash \{s\}}
{\lambda_j 2^{\tilde{R}_j}}}\\
&\qquad\rm{subject\;\;to}\nonumber\\
&\phi\left(-f_{i,j}^{d}\right)\leq 0\;\;\;\forall(i,j)\in\mathcal{E},d\in\mathcal{D}_s\\
&\phi\left(f_{i,j}^d-r_{i,j}\right)\leq 0\;\;\;\forall(i,j)\in\mathcal{E},d\in\mathcal{D}_s\\
\label{con:modified_flow_conservation_upper_side}
&\phi\left(q_j^d\right)\leq 0\qquad\forall j\in\mathcal{V}\backslash \{s\},d\in\mathcal{D}_s\\
\label{con:modified_flow_conservation_lower_side}
&\phi\left(-q_j^d\right)\leq 0\qquad\forall j\in\mathcal{V}\backslash \{s\},d\in\mathcal{D}_s,
\end{align}
\end{subequations}
where $\tilde{R}_j$ was defined in (\ref{eq:r_j_tilde}), $\phi(x)=e^x-1$
and for all $j\in\mathcal{V}\backslash \{s\},d\in\mathcal{D}_s$
\begin{align}
&q_j^d:=\displaystyle\sum_{i\in\mathcal{I}(j)}{f_{i,j}^{d}}-
\displaystyle\sum_{i\in\mathcal{O}(j)}{f_{j,i}^{d}}-\psi_j^d\\
&\psi_j^d:=\begin{cases}
	0 & j\notin \{s,d\}\\
	R_s & j=d
\end{cases}.
\end{align}
Note that theorem 11 in \cite{Feijer_stability_2010} that guarantees convergence for the corresponding dynamic system was proved under the assumption that Slater's condition holds for the modified convex optimization problem. It is easy to see that Slater's condition does not hold for (\ref{opt:MAC_equivalent}). However, it can be verified that their proofs remain valid as they are under any other constraint qualification (i.e., whenever strong duality for the modified optimization holds). In Appendix \ref{appendix:qualification}, we show that strong duality holds for (\ref{opt:MAC_equivalent}). Denote by $\rho_{i,j}^d,w_{i,j}^d,\varphi_j^d$ and $\mu_j^d$ the dual variables of
(\ref{opt:MAC_equivalent}) and define
\begin{equation}
\Delta_{i,j}^d:=-\varphi_j^d e^{q_j^d}+I_{i\neq s}\varphi_i^d e^{q_i^d}+\mu_j^d e^{-q_j^d}-I_{i\neq s}\mu_i^d e^{-q_i^d},
\end{equation}
where $I_{i\neq s}=1$ if $i\neq s$ and zero otherwise. The primal-dual gradient laws for (\ref{opt:MAC_equivalent}) are given by
\begin{align}
\label{dynamic_begin}
&\dot{r}_{i,j}=\tau_{i,j}
\left(-\lambda_j e^{\tilde{R}_j}+\displaystyle\sum_{d\in\mathcal{D}_s}{w_{i,j}^d e^{f_{i,j}^d-r_{i,j}}}\right)\\
&\dot{f}_{i,j}^d=k_{i,j}^d\left(\rho_{i,j}^d e^{-f_{i,j}^d}-w_{i,j}^d e^{f_{i,j}^d-r_{i,j}}+\Delta_{i,j}^d\right)\\
&\dot{\rho}_{i,j}^d=\alpha_{i,j}^d\left[e^{-f_{i,j}^d}-1\right]_{\rho_{i,j}^d}^{+}\\
&\dot{w}_{i,j}^d=\theta_{i,j}^d \left[e^{f_{i,j}^d-r_{i,j}}-1\right]_{w_{i,j}^d}^{+}\\
&\dot{\varphi}_j^d=\beta_j^d\left[e^{q_{j}^d}-1\right]_{\varphi_j^d}^{+}\\
\label{dynamic_end}
&\dot{\mu}_j^d=\gamma_j^d\left[e^{-q_{j}^d}-1\right]_{\mu_j^d}^{+},
\end{align}
where $\tau_{i,j},k_{i,j}^d,\theta_{i,j}^d,\alpha_{i,j}^d,\beta_j^d$ and $\gamma_j^d$ are some positive scalars and for any two scalars $x$ and $p$
\begin{equation}
[x]_{p}^{+}=\begin{cases}
	0 & x<0,p<0 \\
	x & \rm{otherwise}.
\end{cases}
\nonumber
\end{equation}
The dynamic (\ref{dynamic_begin})-(\ref{dynamic_end}) can be distributively implemented by associating a processor for each node in the network, excluding the source node. Each node $j$'s processor keeps track of the variables $\varphi_j^d$ and $\mu_j^d$ as well as the variables associated with node $j$'s ingoing links (i.e., the links in $\{(i,j): i\in\mathcal{I}(j)\}$). Note that message passing is required only between direct neighbors.

\section{Simulation results}
\label{sec:simulation}
In this section the probability of outage of the suboptimal algorithm for the MAC network model is presented. In the simulation we consider the networks shown in Fig. \ref{fig:MAC_network_model} where it was assumed that all links are i.i.d $\rm{Rayleigh}(\upsilon)$ channels, with $\upsilon=1$. We solved (\ref{opt:suboptimal_MAC_convex}) for various values of $\SNR=\frac{P}{\sigma^2}$ and the results are shown in Fig. \ref{fig:resq_example_OutageProbability_VS_MulticastRate_fullPaperNetwork}.
The lower and upper bounds for the outage probability were obtained by calculating (\ref{ineq:Lower_MAC_Outage_Bound_gamma}) and (\ref{ineq:Upper_MAC_Outage_Bound}) (respectively) for each MAC associated with node $j$ if $\mathcal{I}(j)\geq 3$ and the exact expression of the outage probability for each receiver $j$ with $\mathcal{I}(j)\leq 2$. As can be seen, up to $6$ bits/sec/Hz the bounds are quite tight.

We compared the performance of the fading MAC network model to the performance of the non-naive TDMA model. The optimal rate allocation scheme that minimizes the probability of outage for the TDMA based model was derived in \cite{Zanko_Network_2013Submeeted}. The results for the non-naive TDMA model are shown in Fig. \ref{fig:NonNaiveOutageProbability_VS_MulticastRate_fullPaperNetwork}. Although for low multicast rate demands there is no significant gain in preferring the MAC network model over the Non-naive TDMA model, when we have a demand for a high multicast rate the MAC network model outperforms the TDMA based model. To emphasize this result, note that in the non-naive TDMA model (see Fig. \ref{fig:NonNaiveOutageProbability_VS_MulticastRate_fullPaperNetwork}) we achieved a multicast rate of $R_s\approx 5.5$ bits/sec/Hz with probability of outage of $P^{out}\approx 0.1$ with $\SNR=30dB$ whereas we obtained the same results (i.e., $R_s\approx 5.5$, $P^{out}\approx 0.1$) in the MAC network model with $\SNR=25dB$.

Finally, we simulated a discrete time version of the distributed algorithm shown in section \ref{sec:distributed}. In this version we consider time steps $m=1,2,\ldots$ and the derivatives were replaced by differences. The scalars $\tau_{i,j},k_{i,j}^d,\theta_{i,j}^d,\alpha_{i,j}^d,\beta_j^d$ and $\gamma_j^d$ can be thought as step sizes. We did not optimized these step sizes and they were randomly chosen at the initiation of the simulation ($\tau_{i,j}$ and $k_{i,j}^d$ where about $10$ times larger than the other step sizes). During the simulation, we considered the network sown in Fig. \ref{fig:MAC_network_model} and it was assumed that node $5$ was the source node and that the destinations were $\mathcal{D}_s=\{1,4,8,10\}$. The convergence of the algorithm is shown in Fig. \ref{fig:Convergence_of_distributed}.

\begin{figure}[htbp]
	\centering
\includegraphics[width=\textwidth]{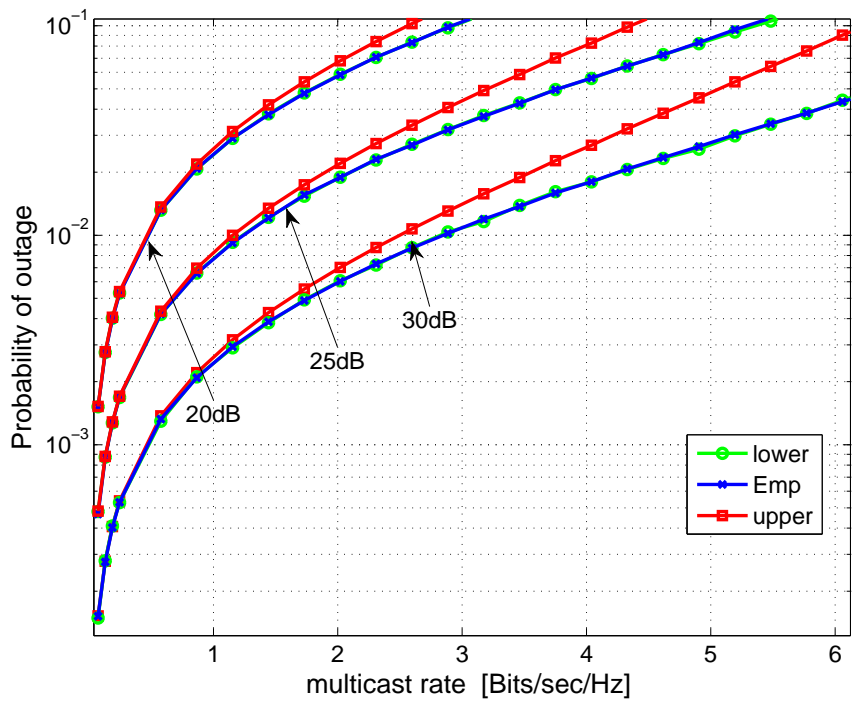}
\caption{The probability of outage of the fading MAC network model for various values of SNR for the network shown in Fig. \ref{fig:MAC_network_model}. It was assumed that node $5$ was the source node and that the destinations were $\mathcal{D}_s=\{1,4,8,10\}$.}
	\label{fig:resq_example_OutageProbability_VS_MulticastRate_fullPaperNetwork}
\end{figure}

\begin{figure}[htbp]
	\centering
\includegraphics[width=\textwidth]{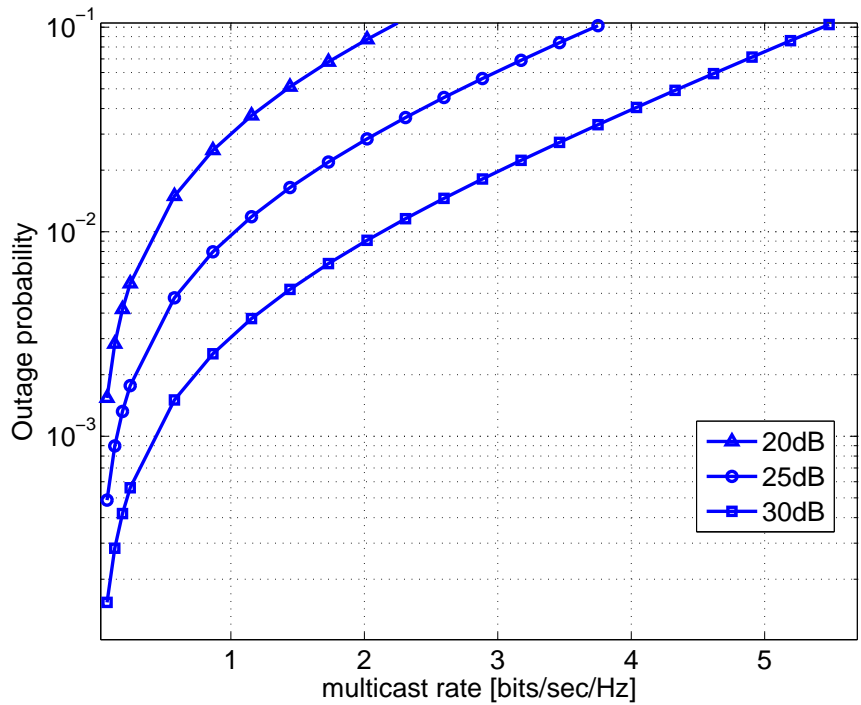}
	\caption{The probability of outage in the non-naive TDMA for various values of SNR for the network shown in Fig. \ref{fig:MAC_network_model}. It was assumed that node $5$ was the source node and that the destinations were $\mathcal{D}_s=\{1,4,8,10\}$.}
	\label{fig:NonNaiveOutageProbability_VS_MulticastRate_fullPaperNetwork}
\end{figure}

\begin{figure}[htbp]
	\centering
\includegraphics[width=\textwidth]{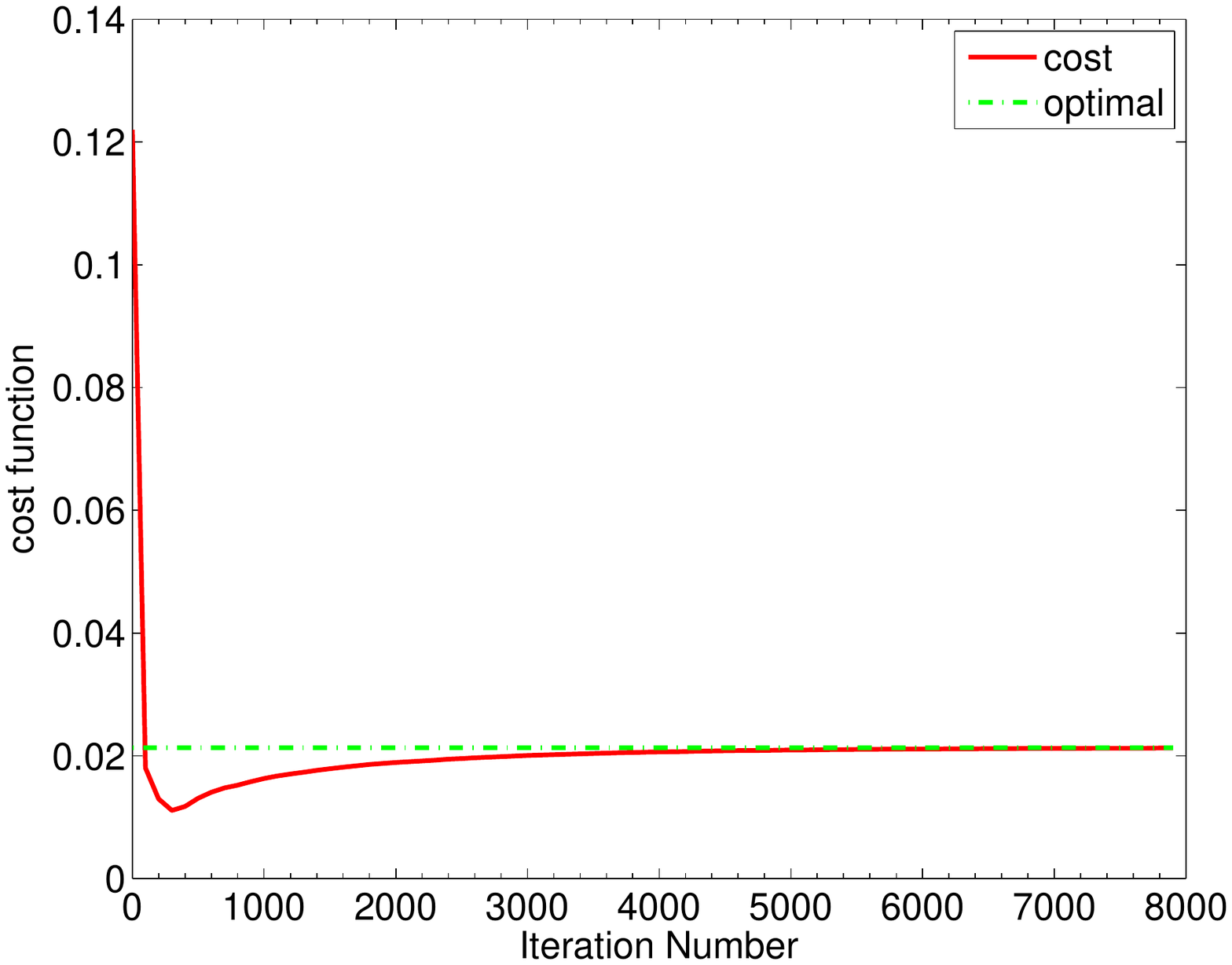}
	\caption{The convergence of the distributed algorithm for the network shown in Fig. \ref{fig:MAC_network_model}. It was assumed that node $5$ was the source node and that the destinations were $\mathcal{D}_s=\{1,4,8,10\}$. The dashed line represents the optimal solution of (\ref{opt:suboptimal_MAC_convex}) and the solid line represents the value of the cost function in (\ref{opt:suboptimal_MAC_convex}) over time.}
	\label{fig:Convergence_of_distributed}
\end{figure}

 \section{Conclusions}
\label{sec:conclusions}
In this paper we studied the rate allocation problem for multicasting over slow Rayleigh fading channels using network coding. A rate allocation scheme based solely on the statistics of the channels was presented.
In the MAC network model, where the network is treated as a collection of slowly Rayleigh fading multiple access channels, we proposed a suboptimal scheme as the solution to a convex optimization problem. This suboptimal solution is based on an upper bound on the probability of outage of a fading multiple access channel. A primal-dual gradient algorithm was derived to solve the problem distributively. In the simulation results, it is shown that the MAC network model outperforms the TDMA based model. The paper provides a practical solution to networks with slow fading channels in which long delays are unacceptable (e.g., in video streaming), with the objective of minimizing outage events throughout the network. As a potential future works, one should consider to derive a bound on the outage probability in the non i.i.d case and extend the problem of statistic-based rate allocation scheme to deal with other than Rayleigh fading model (e.g., Rician or Nakagami models).
\appendix
\section{A recursion for computing the joint distribution of linear combinations of exponential random variables}
\label{appendix:recursion}
In this section we present a new version of a recursion that appeared as equation (17) in \cite{Huffer_Computing_2001}. The lemma in \cite{Huffer_Computing_2001} gives a recursion for the computation of the joint distribution of linear combinations of spacings of uniform distribution. The authors in \cite{Huffer_Computing_2001} remarked that the recursion remain valid as well for computing the joint distribution of linear combinations of exponential random variables. This is inaccurate and we revise the result to handle exponential random variables. To that end, we need the following notations.
Let $\Psi=\{i_1,i_2,\ldots,i_k\}\;\;k\geq 1$ be a set of indices of columns of a matrix $\bf A$ such that $i_{\ell}<i_{\ell+1}$ for all $\ell$ and let ${\bf A}_{-\Psi(m)}$ denote the submatrix of $A$ constructed by deleting the columns of $\bf A$ indexed by $\{i_1,i_2,\ldots,i_m\}$.
\begin{lemma}
\label{lemma:Huffer_computing_2001}
Let $z_1,z_2,\cdots,z_{_{N+1}}$ be $(N+1)$ i.i.d exponential random variables with expectation $E\{z_i\}=1$. Let $\Psi=\{i_1,i_2,\ldots,i_k\}$, $k\geq 1$ be a set of indices of identical columns of matrix $\bf A$ (without loss of generality $i_{\ell}<i_{\ell+1}$ for all $\ell$). If there exists a row $r$ in $\bf A$ such that a) $a_{r,i}>0$ for $i\in\Psi$, b) $a_{r,i}=0$ for $i\notin\Psi$ and c) $b_r\geq 0$, the following recursion holds
\begin{equation}
\label{eq:Huffer_computing_2001}
\Pr\left({\bf Az}>\lambda{\bf b}\right)=\displaystyle\sum_{m=0}^{k-1}{
\frac{1}{m!}(\lambda\delta )^m e^{-\lambda\delta}\Pr\left({\bf A}^{r*}_{-\Psi(m)}{\bf z}>\lambda{\bf c}\right)},
\end{equation}
where $\delta=\frac{b_r}{a_{r,i_1}}$ and ${\bf c}={\bf b}^{r*}-\delta{\bf a}_{i_1}^{r*}$.
\end{lemma}
\begin{proof}
As pointed out in \cite{Huffer_Computing_2001}, since the expression ${\bf Az}>\lambda{\bf b}$ stands for a conjunction of inequalities involving i.i.d random variables, $\Pr\left({\bf Az}>\lambda{\bf b}\right)=\Pr\left({\boldsymbol \pi}{\bf Az}>\lambda{\boldsymbol \pi}{\bf b}\right)$ and $\Pr\left({\bf Az}>\lambda{\bf b}\right)=\Pr\left({\bf A}{\boldsymbol \pi}{\bf z}>\lambda{\boldsymbol \pi}{\bf b}\right)$ hold for any permutation matrix ${\boldsymbol \pi}$ with the appropriate dimensions. Therefore, without loss of generality, we assume that $r=1$ and $\Psi=\{1,2,\ldots,k\}$ (See the illustration of such a matrix in Fig. \ref{fig:R1_R3_A_matrix}).
\begin{figure}[htbp]
	\centering
		\includegraphics[width=0.8\textwidth]{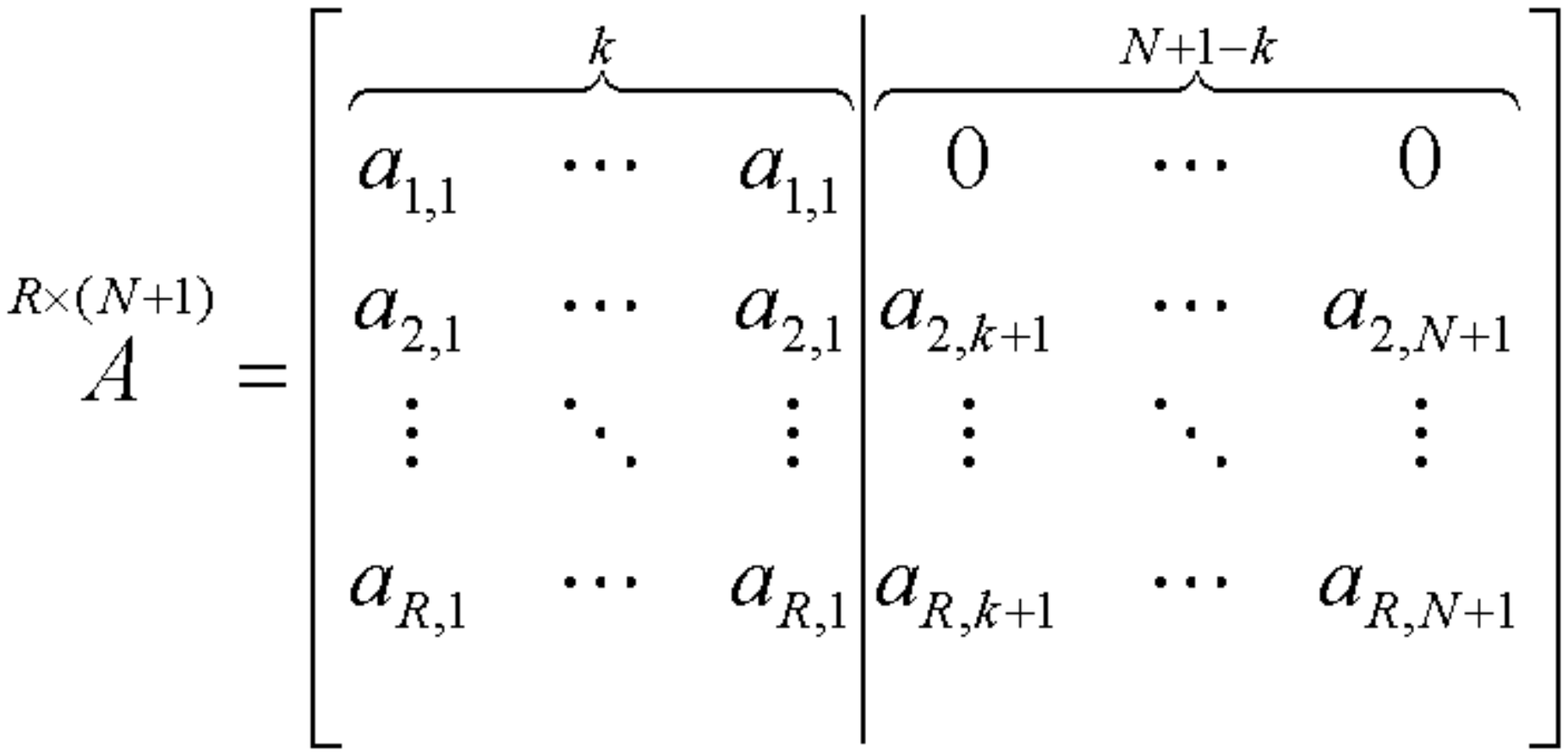}
	\caption{An illustration of a matrix $\bf A$ that satisfies the assumptions in Lemma \ref{lemma:Huffer_computing_2001}, where $\Psi=\{1,2,\ldots,k\}$ and $r=1$.}
	\label{fig:R1_R3_A_matrix}
\end{figure}
Under these assumptions, the first inequality in ${\bf Az}>\lambda{\bf b}$ is
\begin{equation}
D=\{a_{1,1}\displaystyle\sum_{\ell=1}^{k}{z_{\ell}}>\lambda b_1\}.
\end{equation}
Clearly,
\begin{equation}
\Pr\left({\bf Az}>\lambda{\bf b}\right)=\Pr\left(D\cap\{{\bf A}^*{\bf z}> \lambda{\bf b}^*\}\right).
\end{equation}
The event $D$ can be written as the union of disjoint events $D=\displaystyle\cup_{m=0}^{k-1}{D_m}$, where
\begin{equation}
D_m=\{\displaystyle\sum_{\ell=1}^{m}{z_{\ell}}\leq\lambda\delta\leq \displaystyle\sum_{\ell=1}^{m+1}{z_{\ell}}\}.
\end{equation}
Therefore,
\begin{equation}
\label{eq:Law_of_total_probability_for_MAC_events}
\Pr\left({\bf Az}>\lambda{\bf b}\right)=\displaystyle\sum_{m=0}^{k-1}
{\Pr\left(D_m\right)\Pr\left({\bf A}^*{\bf z}>\lambda {\bf b}^*|D_m\right)}.
\end{equation}
For $m<k$, the $r$th inequality in ${\bf Az}>\lambda{\bf b}$ can be rewritten as
\begin{align}
\label{ineq:R1_R3_inequalities}
&\{\displaystyle\sum_{\ell=1}^{N+1}{a_{r,\ell}z_{\ell}>\lambda b_r}\}=
\{\displaystyle\sum_{\ell=1}^{m+1}{a_{r,\ell}z_{\ell}}+
  \displaystyle\sum_{\ell=m+2}^{N+1}{a_{r,\ell}z_{\ell}>\lambda b_r}\}=\nonumber\\
&\{a_{r,1}\displaystyle\sum_{\ell=1}^{m+1}{z_{\ell}}+
  \displaystyle\sum_{\ell=m+2}^{N+1}{a_{r,\ell}z_{\ell} - a_{r,1}\lambda\delta  >\lambda b_r - a_{r,1}\lambda\delta }\}=\nonumber\\
& \{a_{r,1}\left(\displaystyle\sum_{\ell=1}^{m+1}{z_{\ell}}-\lambda\delta\right)+
  \displaystyle\sum_{\ell=m+2}^{N+1}{a_{r,\ell}z_{\ell}} -  >\left(b_r - \lambda\delta a_{r,1}\right)\}=\nonumber\\
&\{\left[a_{r,1},a_{r,m+2},\cdots,a_{_{r,N+1}}\right]^T {\bf T}^m>\left(b_r - \lambda\delta a_{r,1}\right)\},
\end{align}
where
\begin{equation}
{\bf T}^m=\left[\displaystyle\sum_{\ell=1}^{m+1}{z_{\ell}}-\lambda\delta,
z_{m+2},\cdots,z_{N+1}\right]^T.
\end{equation}
Therefore,
\small
\begin{equation}
\label{eq:lower_dimensionality_of_A}
\Pr\left({\bf A}^*{\bf z}>\lambda{\bf b}^*|D_m\right)=
\Pr\left({\bf A}_{(-m)}^*{\bf T}^m>\lambda\left({\bf b}^*-\delta{\bf a}^*\right)|D_m\right)
\end{equation}

\normalsize
In the following we show that the event ${\bf T}^m|D_m$ has the same distribution as $(n+1-m)$ i.i.d exponential random variables. Obviously, $D_m$ is independent with $\left({\bf T}^{m}\right)^*$. Therefore, in order to show that ${\bf T}^m|D_m$ has the same distribution as $(n+1-m)$ i.i.d exponential random variables, it suffices to show that $\displaystyle\sum_{\ell=1}^{m+1}{z_{\ell}}-\lambda\delta |D_m$ has the same distribution as $z_{m+1}$. Note that
\small
\begin{align}
&\left\{\displaystyle\sum_{\ell=1}^{m+1}{z_{\ell}}-\lambda\delta |D_m\right\}=
\left\{\displaystyle\sum_{\ell=1}^{m+1}{z_{\ell}}-\lambda\delta|0\leq \lambda\delta -\displaystyle\sum_{\ell=1}^{m}{z_{\ell}}\leq z_{m+1}\right\}.
\end{align}
\normalsize
Therefore,
\begin{align}
&\Pr\left(\displaystyle\sum_{\ell=1}^{m+1}{z_{\ell}}-\lambda\delta >x|0\leq \lambda\delta -\displaystyle\sum_{\ell=1}^{m}{z_{\ell}}\leq z_{m+1}\right)=\\
\label{eq:before_memory_property}
&\Pr\left(z_{m+1}>x-\displaystyle\sum_{\ell=1}^{m}{z_{\ell}}+\lambda\delta|0\leq \lambda\delta-\displaystyle\sum_{\ell=1}^{m}{z_{\ell}}\leq z_{m+1}\right).
\end{align}
Due to the memoryless property\footnote{The memoryless property of an exponential variable means that for any $a,b\geq 0$, we have that $\Pr\left(Z>a+b|Z>a\right)=\Pr\left(Z>b\right)$, where $Z$ is exponential random variable.} of the exponential random variable we can write
\begin{align}
&\Pr\left(z_{m+1}>x-\displaystyle\sum_{\ell=1}^{m}{z_{\ell}}+\lambda\delta|z_{m+1}\geq \lambda\delta-\displaystyle\sum_{\ell=1}^{m}{z_{\ell}}\geq 0\right)=\\
&\Pr\left(z_{m+1}>x-\displaystyle\sum_{\ell=1}^{m}{z_{\ell}}+\lambda\delta -\left(\lambda\delta-\displaystyle\sum_{\ell=1}^{m}{z_{\ell}}\right)\right)=\\
&\Pr\left(z_{m+1}>x\right).
\end{align}
Therefore, (\ref{eq:lower_dimensionality_of_A}) becomes
\begin{equation}
\label{eq:lower_dimensionality_A}
\Pr\left({\bf A}^*{\bf z}>\lambda{\bf b}^*|D_m\right)=
\Pr\left({\bf A}^*_{(-m)}{\bf z}>\lambda\left({\bf b}^*-\delta{\bf a}^*\right)\right).
\end{equation}
Combining (\ref{eq:Law_of_total_probability_for_MAC_events}) and (\ref{eq:lower_dimensionality_A}) yields
\begin{equation}
\Pr\left({\bf Az}>\lambda{\bf b}\right)=\displaystyle\sum_{m=0}^{k-1}
{\Pr\left(D_m\right)\Pr\left({\bf A}^*_{(-m)}{\bf z}>\lambda\left({\bf b}^*-\delta{\bf a}^*\right)\right)}.
\end{equation}
In order to complete the recursion we need an explicit expression of $\Pr(D_m)$.
\begin{align}
\Pr(D_m)&=\Pr\left(\displaystyle\sum_{\ell=1}^{m}{z_{\ell}}\leq\lambda\delta\leq \displaystyle\sum_{\ell=1}^{m+1}{z_{\ell}}\right)\\
&=1-\Pr\left(\displaystyle\sum_{\ell=1}^{m}{z_{\ell}}>\lambda\delta  \right)-\Pr\left(\lambda\delta> \displaystyle\sum_{\ell=1}^{m+1}{z_{\ell}} \right)\\
&=\Pr\left(\displaystyle\sum_{\ell=1}^{m}{z_{\ell}}\leq\lambda\delta \right)-\Pr\left( \displaystyle\sum_{\ell=1}^{m+1}{z_{\ell}}\leq \lambda\delta  \right).
\end{align}
Since $\displaystyle\sum_{\ell=1}^{m}{z_{\ell}}$ follows the $\rm{Erlang}(m,1)$ distribution\footnote{The CDF of Erlang$(m,\lambda)$ distributed random variable $Y$ is given by $\;\;\;F_Y(y)=1-\displaystyle\sum_{\ell=0}^{m-1}{\frac{1}{\ell!}(\lambda y)^{\ell} e^{-\lambda y}}$.}, we have that
\begin{equation}
\Pr(D_m)=\frac{1}{m!}(\lambda\delta)^m e^{-\lambda\delta}.
\end{equation}
The claim now follows.
\end{proof}

\section{Strong duality of (\ref{opt:MAC_equivalent})}
\label{appendix:qualification}
\begin{lemma}
\label{lemma:modified_slater}
Strong duality holds for (\ref{opt:MAC_equivalent}).
\end{lemma}
\begin{proof}
Consider the following optimization problem
\begin{subequations}
\label{opt:MAC_equivalent_modified}
\begin{align}
\tag{\ref{opt:MAC_equivalent_modified}}
&\displaystyle\min_{{\bf f},{\bf r}}
{\displaystyle\sum_{j\in\mathcal{V}\backslash \{s\}}
{\lambda_j 2^{\tilde{R}_j}}}\\
&\qquad\rm{subject\;\;to}\nonumber\\
&\phi\left(-f_{i,j}^{d}\right)\leq 0\;\;\;\forall(i,j)\in\mathcal{E},d\in\mathcal{D}_s\\
&\phi\left(f_{i,j}^d-r_{i,j}\right)\leq 0\;\;\;\forall(i,j)\in\mathcal{E},d\in\mathcal{D}_s\\
&\displaystyle\sum_{i\in\mathcal{I}(j)}{f_{i,j}^{d}}-
\displaystyle\sum_{i\in\mathcal{O}(j)}{f_{j,i}^{d}}=\begin{cases}
												0 & j\notin \{s,d\}\\
												R_s & j=d\\
											\end{cases}\\
&\qquad\qquad\qquad\qquad\qquad\forall j\in\mathcal{V}\backslash \{s\},d\in\mathcal{D}_s.\nonumber
\end{align}
\end{subequations}
Note that the refined Slater's condition holds for (\ref{opt:MAC_equivalent_modified}) and therefore (\ref{opt:MAC_equivalent_modified}) has zero duality gap \cite{Boyd_Convex_2004}, but it does not hold for (\ref{opt:MAC_equivalent}). Obviously, the feasible sets of (\ref{opt:suboptimal_MAC_convex}), (\ref{opt:MAC_equivalent}) and (\ref{opt:MAC_equivalent_modified}) are all the same and therefore they have identical optimal solutions. We need to show that solving (\ref{opt:MAC_equivalent}) through its Lagrangian yields the same solution. Denote the primal variables by ${\bf x}=({\bf f},{\bf r})$ and the dual variables by $\boldsymbol{\zeta}$. Denote the Lagrangians of (\ref{opt:MAC_equivalent}) and (\ref{opt:MAC_equivalent_modified}) by $L({\bf x},\boldsymbol{\zeta})$, $L_M({\bf x},\boldsymbol{\zeta})$, respectively. The dual function of (\ref{opt:MAC_equivalent}) is given by $q(\boldsymbol{\zeta}) =
\displaystyle\min_{{\bf x}}{L({\bf x},\boldsymbol{\zeta})}$. Assume that there exists ${\bf \tilde{x}}(\boldsymbol{\zeta})=({\bf \tilde{f}},{\bf \tilde{r}})$ that is a minimizer of
$L({\bf x},\boldsymbol{\zeta})$ that does not obey the flow conservation constraint (\ref{con:flow_conservation_MAC_convex}). Therefore, there exists node $j\in\mathcal{V}\backslash \{s\}$ such that $q_j^d\neq 0$. Hence, we have either $\phi(q_j^d)>0$ or $\phi(-q_j^d)>0$. Without loss of generality assume that $\phi(q_j^d)>0$. In that case, since the cost function in (\ref{opt:MAC_equivalent}) is bounded below by $0$ we can always choose $\boldsymbol{\zeta}$ such that the dual solution $q(\boldsymbol{\zeta})$ is infinity (by setting all $\lambda_i$ to zero except the one with the positive coefficient $\phi(q_j^d)>0$). This contradicts the feasibility of the primal (\ref{opt:MAC_equivalent}). We conclude by noting that $L({\bf x},\boldsymbol{\zeta})=L_M({\bf x},\boldsymbol{\zeta})$ for any ${\bf x}$ obeys the flow conservation constraint (\ref{con:flow_conservation_MAC_convex}).
\end{proof}
\bibliographystyle{model1-num-names}

\end{document}